\documentclass[letterpaper, 10 pt, conference]{ieeeconf}  

\IEEEoverridecommandlockouts                              

\usepackage{xcolor}
\def\BibTeX{{\rm B\kern-.05em{\sc i\kern-.025em b}\kern-.08em
    T\kern-.1667em\lower.7ex\hbox{E}\kern-.125emX}}

\usepackage{textgreek} 
\usepackage{amsmath,amssymb,amsfonts,mathrsfs,commath,mathtools}
\usepackage{algorithmic}
\usepackage{algorithm}
\usepackage{array}
\usepackage[caption=false,font=normalsize,labelfont=sf,textfont=sf]{subfig}
\usepackage{textcomp}
\usepackage{stfloats}
\usepackage{url}
\usepackage{verbatim}
\usepackage{graphicx}
\usepackage{cite}
\usepackage{tabularx}

\usepackage{setspace}
\usepackage{subfloat}
\usepackage{mathtools}
\usepackage{amsmath,amsfonts,amssymb,mathrsfs, bm}
\usepackage[utf8]{inputenc}
\usepackage[swedish, english]{babel}
\usepackage{csquotes}
\usepackage{subfig}
\usepackage{graphicx}
\usepackage{supertabular}                                             
\usepackage{booktabs}

\allowdisplaybreaks

\DeclareMathOperator{\RE}{Re}

\DeclareMathOperator{\esup}{ess\, sup}



\providecommand*{\eu}%
{\ensuremath{\mathrm{e}}}
\providecommand*{\im}%
{\ensuremath{\mathrm{j}}}
\providecommand*{\GammaF}%
{\ensuremath{\mathrm{\Gamma}}}
\providecommand*{\BetaF}%
{\ensuremath{\mathrm{\Beta}}}
\usepackage{ulem}

\newtheorem{lemma}{Lemma}[section]
\newtheorem{theorem}{Theorem}[section]

\usepackage[table]{xcolor} 
\usepackage{hyperref}
\hypersetup{
    colorlinks = true,
    linkbordercolor = {white},
            linkcolor=blue,
            bookmarks=true,
            filecolor=blue,
            urlcolor=blue,
            citecolor=blue
}

\title{\LARGE \bf
Inverse-dynamics observer design for a linear single-track\\ vehicle model with distributed tire dynamics}

\author{Luigi Romano$^{1}$, Ole Morten Aamo$^{2}$, Jan Åslund$^{3}$, and Erik Frisk$^{3}$
\thanks{This research was financially supported by the project FASTEST (Reg. no. 2023-06511), funded by the Swedish Research Council.}
\thanks{$^{1}$Luigi Romano is with the Department of Electrical Engineering, Linköping University, Linköping, Sweden, and the Department of Engineering Cybernetics, NTNU, Trondheim, Norway
        {\tt\small luigi.romano@liu.se}}%
\thanks{$^{2}$Ole Morten Aamo is with the Department of Engineering Cybernetics, NTNU, Trondheim, Norway
        {\tt\small aamo@ntnu.no}}%
\thanks{$^{3}$Jan Åslund and Erik Frisk are with the Department of Electrical Engineering, Linköping University, Linköping
       {\tt\small jan.aslund@liu.se}; {\tt\small erik.frisk@liu.se}}%
}

\begin{document}

\maketitle
\thispagestyle{empty}
\pagestyle{empty}


\begin{abstract}
Accurate estimation of the vehicle's sideslip angle and tire forces is essential for enhancing safety and handling performances in unknown driving scenarios. To this end, the present paper proposes an innovative observer that combines a linear single-track model with a distributed representation of the tires and information collected from standard sensors. In particular, by adopting a comprehensive representation of the tires in terms of hyperbolic partial differential equations (PDEs), the proposed estimation strategy exploits dynamical inversion to reconstruct the lumped and distributed vehicle states solely from yaw rate and lateral acceleration measurements. Simulation results demonstrate the effectiveness of the observer in estimating the sideslip angle and tire forces even in the presence of noise and model uncertainties.
\end{abstract}


\section{Introduction}
\label{sec:introduction}
Accurate estimation of vehicle sideslip angle and tire forces is essential for enhancing stability, safety, and control performance in modern automotive systems \cite{Lex1,Lex2,Sideslip1,Sideslip2,Barys1}. Reliable state estimation enables advanced driver-assistance and autonomous driving functions to maintain optimal performance even under uncertain road and operating conditions \cite{Barys2,Emilia1,Emilia2}. However, the distributed nature of the tire–road interaction introduces significant modeling and estimation challenges, particularly during transient maneuvers where conventional lumped tire models fail to capture spatially distributed friction dynamics.

Inspired by the design strategies presented in \cite{AuriolX1,AuriolX2,AuriolX3,AuriolX4,InversionMio}, this paper proposes an inverse-dynamics observer for a linear single-track vehicle model incorporating distributed tire dynamics. The tire behavior is represented through linear hyperbolic partial differential equations (PDEs), which are coupled with the ordinary differential equation (ODE) describing the vehicle's rigid body dynamics, resulting in a homodirectional ODE-PDE interconnection \cite{SemilinearV}. By inverting the dynamics of the PDE subsystem, which is inherently stable \cite{TireNonDyn}, the proposed observer reconstructs the lumped and distributed vehicle states using only standard sensor measurements of yaw rate and lateral acceleration. In particular, theoretical analysis establishes stability conditions through frequency-domain arguments, and the proposed injection gain is designed to guarantee robust convergence in the presence of model uncertainties.

The observer's performance is tested through numerical simulations, confirming that the proposed approach can accurately estimate the sideslip angle and the tire forces, maintaining robustness against sensor noise and parameter variations. Furthermore, simulation results illustrate the ability of the observer to reconstruct the distributed tire states, which cannot be achieved when traditional lumped models are employed.

The paper bridges a novel vehicular application with recent advances in control and estimation of distributed parameter systems. To the best of the authors’ knowledge, it represents the first contribution proposing an observer capable of reconstructing both lumped and distributed vehicle states using only standard on-board measurements.

The remainder of this manuscript is organized as follows: Section \ref{sect:Problem} introduces the vehicle model considered in the paper. The state observer for vehicle state estimation is then synthesized in \ref{sect:observerDesign}. Section \ref{sect:example} presents numerical experiments, whereas Section \ref{ref:concl} concludes with final remarks.

\subsection{Notation}\label{sect:Notation}
In this paper, $\mathbb{R}$ denotes the set of real numbers; $\mathbb{R}_{>0}$ and $\mathbb{R}_{\geq 0}$ indicate the set of positive real numbers excluding and including zero, respectively. Similarly, $\mathbb{C}$ is the set of complex numbers; $\mathbb{C}_{>0}$ and $\mathbb{C}_{\geq 0}$ denote the sets of all complex numbers whose real part is larger than and larger than or equal to zero, respectively. The set of integers is denoted by $\mathbb{Z}$.
The group of $n\times m$ matrices with values in $\mathbb{F}$ ($\mathbb{F} = \mathbb{R}$, $\mathbb{R}_{>0}, \mathbb{R}_{\geq0}$, or $\mathbb{C}$) is denoted by $\mathbf{M}_{n\times m}(\mathbb{F})$ (abbreviated as $\mathbf{M}_{n}(\mathbb{F})$ whenever $m=n$). For $\mathbb{F} = \mathbb{R}$, $\mathbb{R}_{>0}$, or $\mathbb{R}_{\geq0}$, $\mathbf{GL}_n(\mathbb{F})$ represents the group of invertible matrices with values in $\mathbb{F}$; the identity matrix on $\mathbb{R}^n$ is indicated with $\mathbf{I}_n$.
The standard Euclidean norm on $\mathbb{R}^n$ is indicated with $\norm{\cdot}_{2}$; matrix norms are simply denoted by $\norm{\cdot}$.
$L^2((0,1);\mathbb{R}^n)$ denotes the Hilbert space of square-integrable functions on $(0,1)$ with values in $\mathbb{R}^n$, endowed with inner product $\langle \bm{\zeta}_1, \bm{\zeta}_2 \rangle_{L^2((0,1);\mathbb{R}^n)} = \int_0^1 \bm{\zeta}_1^{\mathrm{T}}(\xi)\bm{\zeta}_2(\xi) \dif \xi$ and induced norm $\norm{\bm{\zeta}(\cdot)}_{L^2((0,1);\mathbb{R}^n)}$. The Hilbert space $H^1((0,1);\mathbb{R}^n)$ indicates the space of square-integrable functions on $(0,1)$ whose weak derivative also belongs to $L^2((0,1);\mathbb{R}^{n})$; it is naturally equipped with norm $\norm{\bm{\zeta}(\cdot)}_{H^1((0,1);\mathbb{R}^n)}^2 \triangleq \norm{\bm{\zeta}(\cdot)}_{L^2((0,1);\mathbb{R}^n)}^2 + \norm{\pd{\bm{\zeta}(\cdot)}{\xi}}_{L^2((0,1);\mathbb{R}^n)}^2$. 
Given a generic Hilbert space $\mathcal{V}$, $L^2((0,T);\mathcal{V})$ and $C^0([0,T];\mathcal{V})$ denote respectively the spaces of square-integrable functions and continuous functions on $[0,T]$ with values in $\mathcal{V}$ (for $T = \infty$, the interval $[0,T]$ is identified with $[0,\infty)$).
The following notation is adapted from \cite{Zwarth,CurtainAutomatica}. Consider $C^\omega(\mathbb{C}_{>0};\mathbf{M}_{n\times m}(\mathbb{C}))$, the set of all analytic functions from $\mathbb{C}_{>0}$ to $\mathbb{C}^{n\times m}$; the Hardy space $H^\infty(\mathbb{C}_{>0};\mathbf{M}_{n\times m}(\mathbb{C}))$ is defined as $H^\infty(\mathbb{C}_{>0};\mathbf{M}_{n\times m}(\mathbb{C})) \triangleq \{ \mathbf{G} \in C^\omega(\mathbb{C}_{>0};\mathbf{M}_{n\times m}(\mathbb{C})) \mathrel{|} \norm{\mathbf{G}(\cdot)}_\infty < \infty\}$, where, for a proper and stable transfer function, $\mathbf{G} \in H^\infty(\mathbb{C}_{>0};\mathbf{M}_{n\times m}(\mathbb{C}))$, $\norm{\mathbf{G}(\cdot)}_\infty = \esup_{\omega \in \mathbb{R}}\overline{\sigma}(\mathbf{G}(\im \omega))$, with $\overline{\sigma}(\mathbf{G}(\im \omega))$ standing for the largest singular value of $\mathbf{G}(\im \omega)$ evaluated at $\omega$. Finally, the Laplace transform of a variable $v(t) \in \mathcal{V}$ is denoted by $\widehat{v}(s) = (\mathcal{L}v)(s)$, where $s \in \mathbb{C}$ is the Laplace variable.

\section{Problem statement}\label{sect:Problem}
This section is dedicated to introducing the considered ODE-PDE system, along with the available measurements.
\subsection{Structure of the considered ODE-PDE system}
\begin{figure}
\centering
\includegraphics[width=0.85\linewidth]{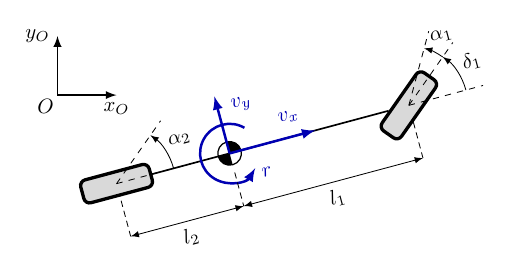} 
\caption{Schematic of the single track model.}
\label{figureForcePostdoc2}
\end{figure}
The model considered in this paper may be derived by linearizing one of the formulations presented in \cite{SemilinearV}, but it is also equivalent to that reported in \cite{MioVSD}.
By adopting a distributed Dahl friction model, the lateral dynamics of the linear single-track vehicle illustrated in Figure \ref{figureForcePostdoc2} may be described by the following homodirectional hyperbolic ODE-PDE interconnection \cite{SemilinearV}:
\begin{subequations}\label{eq:originalSystems}
\begin{align}
\begin{split}
& \dot{\bm{X}}(t) = \mathbf{A}_1\bm{X}(t) +\mathbf{G}(\mathscr{K}_1\bm{z})(t),\quad  t \in (0,T), 
\end{split} \label{eq:originalSystemsODE}\\
\begin{split}
& \dpd{\bm{z}(\xi,t)}{t} + \mathbf{\Lambda} \dpd{\bm{z}(\xi,t)}{\xi} = (\mathscr{K}_2\bm{z})(t) + \mathbf{A}_2\bm{X}(t) +\mathbf{B}\bm{\delta}(t),\\
& \qquad \qquad  \qquad  \qquad  \qquad  \quad (\xi,t) \in (0,1) \times (0,T),
\end{split} \label{eq:originalSystemsPDE} \\
& \bm{z}(0,t) = \bm{0}, \quad t \in (0,T),\label{eq:originalSystemsBC}
\end{align}
\end{subequations}
where $\mathbb{R}^{2} \ni \bm{X}(t) = [v_y(t)\; r(t)]^{\mathrm{T}}$ represents the lumped state vector, collecting the vehicle's lateral speed $v_y(t) \in \mathbb{R}$ and yaw rate $r(t) \in \mathbb{R}$, $\mathbb{R}^2 \ni \bm{z}(\xi,t) \triangleq [z_1(\xi,t) \; z_2(\xi,t)]^{\mathrm{T}}$ indicates the distributed state vector, describing the lateral deformation of bristles attached at the front and rear tires and schematizing tread elements, and the input $\mathbb{R}^2\ni \bm{\delta}(t) \triangleq [\delta_1(t)\; \delta_2(t)]^{\mathrm{T}}$ collects the steering angles at the tires, and the variable $\xi \in [0,1]$ is the nondimensional coordinate along the tires' contact patch lengths \cite{SemilinearV}. 

The matrices are $\mathbf{A}_1,\mathbf{A}_2, \mathbf{G}, \mathbf{B} \in \mathbf{M}_2(\mathbb{R})$ and $\mathbf{GL}_2(\mathbb{R}) \cap \mathbf{M}_2(\mathbb{R}) \ni \mathbf{\Lambda} \succ \mathbf{0}$ are given by
\begin{align}\label{eq:matrices}
\mathbf{A}_1 & \triangleq \begin{bmatrix} 0 & -v_x \\ 0 & 0 \end{bmatrix}, && \mathbf{A}_2 \triangleq 2\begin{bmatrix} \phi_1 & \phi_1l_1 \\ \phi_2 & -\phi_2l_2\end{bmatrix}, \nonumber \\
\mathbf{G} & \triangleq -\begin{bmatrix} \dfrac{1}{m} & \dfrac{1}{m} \\ \dfrac{l_1}{I_z} & -\dfrac{l_2}{I_z}\end{bmatrix}, && \mathbf{B} \triangleq-2v_x\begin{bmatrix} \phi_1 & 0 \\ 0 & \chi \phi_2\end{bmatrix}, \nonumber \\
\mathbf{\Lambda} & \triangleq \begin{bmatrix} \dfrac{v_x}{L_1} & 0 \\ 0 & \dfrac{v_x}{L_2}\end{bmatrix},
\end{align}
whereas the operators $(\mathscr{K}_1\bm{\zeta})$ and $(\mathscr{K}_2\bm{\zeta})$ satisfy $\mathscr{K}_1 \in \mathscr{L}(L^2((0,1);\mathbb{R}^2);\mathbb{R}^2)$ and $\mathscr{K}_2 \in \mathscr{L}(H^1((0,1);\mathbb{R}^2);\mathbb{R}^2)$, and read
\begin{align}\label{eq:operatorsK}
(\mathscr{K}_1\bm{\zeta}) & \triangleq \int_0^1 \mathbf{K}_1\bm{\zeta}(\xi)\dif \xi, &&  (\mathscr{K}_2\bm{\zeta}) \triangleq \mathbf{K}_2\bm{\zeta}(1),
\end{align}
with $\mathbf{K}_1,\mathbf{K}_2 \in \mathbf{M}_2(\mathbb{R})$ given by
\begin{align}\label{eq:matricesK}
\mathbf{K}_1 & \triangleq \begin{bmatrix} F_{z1}\sigma_1 & 0 \\ 0 & F_{z2}\sigma_2\end{bmatrix}, && \mathbf{K}_2 \triangleq v_x\begin{bmatrix} \dfrac{\psi_1}{L_1} & 0 \\ 0 & \dfrac{\psi_2}{L_2} \end{bmatrix}.
\end{align}
In \eqref{eq:matrices} and \eqref{eq:matricesK}, $v_x\in \mathbb{R}_{>0}$ represents the constant longitudinal speed of the vehicle, $m\in \mathbb{R}_{>0}$ and $I_z\in \mathbb{R}_{>0}$ denote respectively the vehicle mass and moment of inertia of the center of gravity around the vertical axis, $l_1, l_2 \in \mathbb{R}_{>0}$ are the front and rear axle lengths, $F_{z1}, F_{z2} \in \mathbb{R}_{>0}$ the vertical forces acting on the front and rear tires, $L_1,L_2 \in \mathbb{R}_{>0}$ the lengths of the front and rear tires' contact patches, and $\sigma_1,\sigma_2\in \mathbb{R}_{>0}$ are the normalized micro-stiffnesses of the front and rear tires. Finally, the constants $\phi_i \in (0,1]$ and $\psi_i \in [0,1)$ are structural parameters connected with the flexibility of the tire carcass, identically satisfying $\phi_i + \psi_i = 1$, $i \in \{1,2\}$, whereas $\chi \in \{0,1\}$ denotes the steering actuation at the rear axle \cite{SemilinearV}.

Well-posedness for the ODE-PDE system \eqref{eq:originalSystems} follows from an application of the Lumer-Phillips theorem (see, for instance, Theorem 1.4.3 and Corollary 4.2.5. in \cite{Pazy}). In particular, the Hilbert spaces $\mathcal{X}\triangleq \mathbb{R}^{2}\times L^2((0,1);\mathbb{R}^{2})$ and $\mathcal{Y} \triangleq \mathbb{R}^{2}\times H^1((0,1);\mathbb{R}^2)$ are considered, equipped respectively with norms $\norm{(\bm{Z}, \bm{\zeta}(\cdot))}_{\mathcal{X}}^2 \triangleq \norm{\bm{Z}}_{2}^2 + \norm{\bm{\zeta}(\cdot)}_{L^2((0,1);\mathbb{R}^{2})}^2$ and $\norm{(\bm{Z}, \bm{\zeta}(\cdot))}_{\mathcal{Y}}^2 \triangleq \norm{\bm{Z}}_{2}^2 + \norm{\bm{\zeta}(\cdot)}_{H^1((0,1);\mathbb{R}^{2})}^2$

\begin{theorem}[Global well-posedness]\label{thm:Ulin}
The ODE-PDE system \eqref{eq:originalSystems} admits a unique \emph{mild solution} $(\bm{X},\bm{z}) \in C^0([0,T];\mathcal{X})$ for all initial conditions (ICs) $(\bm{X}_0,\bm{z}_0) \triangleq (\bm{X}(0), \bm{z}(\cdot,0)) \in \mathcal{X}$ and inputs $\bm{\delta} \in L^p((0,T);\mathbb{R}^2)$, $p\geq 1$. If, in addition, $(\bm{X}_0,\bm{z}_0) \in \mathcal{Y}$ satisfies the boundary condition (BC) \eqref{eq:originalSystemsBC}, and $\bm{\delta} \in C^1([0,T];\mathbb{R}^2)$, the solution is \emph{classical}, that is, $(\bm{X}, \bm{z}) \in C^1([0,T];\mathcal{X})\cap C^0([0,T];\mathcal{Y})$ and satisfies the BC \eqref{eq:originalSystemsBC}.

\begin{proof}
See Theorem 4.1 in \cite{SemilinearV}.
\end{proof}
\end{theorem}
For simplicity, this paper considers classical solutions, which ensure that the boundary terms are well-defined and allow the avoidance of overly technical arguments. A schematic of the ODE-PDE system \eqref{eq:originalSystems} is illustrated in Figure \ref{figure:system0}, where, for convenience of notation, $\mathbb{R}^{2} \ni \bm{W}_1(t) \triangleq \mathbf{G}(\mathscr{K}_1\bm{z})(t)$ and $\mathbb{R}^{2} \ni \bm{W}_2(t) \triangleq \mathbf{A}_2\bm{X}(t) + \mathbf{B}\bm{\delta}(t)$.

\begin{figure}
\centering
\includegraphics[width=0.9\linewidth]{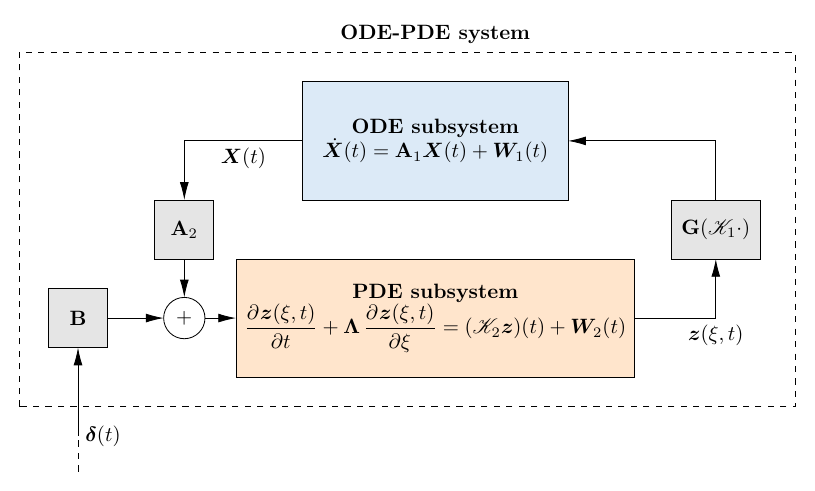} 
\caption{Schematic representation of the ODE-PDE interconnection \eqref{eq:originalSystems}.}
\label{figure:system0}
\end{figure}

The objective of this work consists of reconstructing the variables $\bm{X}(t)$ and $\bm{z}(\xi,t)$ from partial knowledge of the lumped and distributed states. Specifically, in what follows, the measurement output $\mathbb{R}^2 \ni \bm{Y}(t) = [Y_1(t) \; Y_2(t)]^{\mathrm{T}} \triangleq [r(t)\; a_y(t)]^{\mathrm{T}}$ is assumed to be of the form
\begin{align}\label{eq:outputY}
Y_1(t) & \triangleq \mathbf{C}_1\bm{X}(t), && Y_2(t) \triangleq \mathbf{C}_2(\mathscr{K}_1\bm{z})(t),
\end{align}
with $\mathbf{C}_1,\mathbf{C}_2 \in \mathbf{M}_{1\times 2}(\mathbb{R})$ reading
\begin{align}
\mathbf{C}_1& \triangleq \begin{bmatrix} 0 & 1\end{bmatrix}, && \textnormal{and} && \mathbf{C}_2 \triangleq -\begin{bmatrix}\dfrac{1}{m} & \dfrac{1}{m} \end{bmatrix}.
\end{align}
In this context, it is worth noting that the ODE subsystem \eqref{eq:originalSystemsODE}, when considered in isolation with measurement $Y_1(t)$, is not detectable. Therefore, detectability should be checked considering the whole interconnections, using, e.g., spectral methods. To avoid explicit infinite-dimensional detectability checks, this paper proposes an inverse-dynamics-based strategy that enables reconstruction of $\bm{X}(t)$ and $\bm{z}(\xi,t)$ without resorting to spectral theory.
Finally, by estimating the lateral velocity $v_y(t)$, the sideslip angle is then obtained directly from $\mathbb{R}\ni \beta(t) \triangleq v_y(t)/v_x$.

\section{Observer design}\label{sect:observerDesign}
The present section constitutes the core of the manuscript and is devoted to the synthesis of a state observer via dynamical inversion.
\subsection{Observer structure}
Denoting with $\hat{\bm{X}}(t) \in \mathbb{R}^{2}$, $\hat{\bm{z}}(\xi,t) \in \mathbb{R}^{2}$, and $\hat{\bm{Y}}(t) \in \mathbb{R}^{2}$ the estimates of $\bm{X}(t)$, $\bm{z}(\xi,t)$, and $\bm{Y}(t)$, respectively, the following observer structure is proposed:
\begin{subequations}\label{eq:obserStruc}
\begin{align}
\begin{split}
& \dot{\hat{\bm{X}}}(t) = \mathbf{A}_1\hat{\bm{X}}(t) +\mathbf{G}(\mathscr{K}_1\hat{\bm{z}})(t)-\mathbf{L}\tilde{Y}_1(t) \\
& \qquad \quad - \Bigl(\mathcal{H}*\tilde{\bm{Y}}\Bigr)(t), \quad  t \in (0,T), 
\end{split} \label{eq:ObsODE}\\
\begin{split}
& \dpd{\hat{\bm{z}}(\xi,t)}{t} + \mathbf{\Lambda} \dpd{\hat{\bm{z}}(\xi,t)}{\xi} = (\mathscr{K}_2\hat{\bm{z}})(t) + \mathbf{A}_2\hat{\bm{X}}(t)+ \mathbf{B}\bm{\delta}(t),\\
& \qquad \qquad \qquad  \qquad  \qquad  \quad (\xi,t) \in (0,1) \times (0,T),
\end{split} \\
& \hat{\bm{z}}(0,t) = \bm{0}, \quad t \in (0,T),
\end{align}
\end{subequations}
where $\mathbb{R}^{2} \ni \tilde{\bm{Y}}(t) = [\tilde{Y}_1(t) \; \tilde{Y}_2(t)]^{\mathrm{T}} \triangleq \bm{Y}(t)-\hat{\bm{Y}}(t)$, the matrix $\mathbf{L}\in \mathbf{M}_{2\times 1}(\mathbb{R})$ is an observer gain to be selected, and $(\mathcal{H}*\tilde{\bm{Y}})(t) \triangleq \int_0^t \mathcal{H}(t-t^\prime)\tilde{\bm{Y}}(t^\prime)\dif t^\prime$ is a stable linear operator to be appropriately defined.

Indicating the lumped and distributed errors with $\mathbb{R}^{2} \ni \tilde{\bm{X}}(t) \triangleq \bm{X}(t)-\hat{\bm{X}}(t)$ and $\mathbb{R}^{2} \ni \tilde{\bm{z}}(\xi,t) \triangleq \bm{z}(\xi,t) -\hat{\bm{z}}(\xi,t)$, subtracting \eqref{eq:obserStruc} from \eqref{eq:originalSystems}, and specifying $\mathbf{L} \triangleq [v_x \; 0]^{\mathrm{T}}$ to eliminate the lumped variables from the right-hand side of the resulting ODE subsystem, the following ODE-PDE system is obtained describing the error dynamics:
\begin{subequations}\label{eq:ErrorDyn}
\begin{align}
\begin{split}
& \dot{\tilde{\bm{X}}}(t) = \mathbf{G}(\mathscr{K}_1\tilde{\bm{z}})(t) + \Bigl(\mathcal{H}*\tilde{\bm{Y}}\Bigr)(t), \quad  t \in (0,T), 
\end{split} \label{eq:ErrorDynODE} \\
\begin{split}
& \dpd{\tilde{\bm{z}}(\xi,t)}{t} + \mathbf{\Lambda} \dpd{\tilde{\bm{z}}(\xi,t)}{\xi} = (\mathscr{K}_2\tilde{\bm{z}})(t)+\mathbf{A}_2\tilde{\bm{X}}(t),\\
& \qquad \qquad  \qquad  \qquad  \qquad  \quad (\xi,t) \in (0,1) \times (0,T),
\end{split}\label{eq:ErrorDynPDE} \\
& \tilde{\bm{z}}(0,t) = \bm{0}, \quad t \in (0,T). \label{eq:ErrorDynBC}
\end{align} 
\end{subequations}
In components, the output error reads 
\begin{align}\label{eq:errorY}
\tilde{Y}_1(t) & \triangleq \mathbf{C}_1\tilde{\bm{X}}(t), && \tilde{Y}_2(t) \triangleq \mathbf{C}_2(\mathscr{K}_1\tilde{\bm{z}})(t).
\end{align}
The goal consists of appropriately designing the operator $(\mathcal{H}*\tilde{\bm{Y}})(t)$, and hence the injection gain $\mathcal{H}(t)$, so that the error dynamics converges exponentially to zero. To this end, it suffices to ensure the exponential convergence of $\tilde{\bm{X}}(t)$.
\begin{lemma}\label{lemma:2}
Consider the error dynamics \eqref{eq:ErrorDyn} with IC $(\tilde{\bm{X}}_0, \tilde{\bm{z}}_0) \in \mathcal{Y}$ satisfying the BC \eqref{eq:ErrorDynBC}, and suppose that $\tilde{\bm{X}}(t) \to \bm{0}$ exponentially fast; then, $\norm{(\tilde{\bm{X}}(t), \tilde{\bm{z}}(\cdot,t))}_{\mathcal{X}} \to 0$ exponentially fast.

\begin{proof}
From Lemma 3.2 in \cite{TireNonDyn}, the unbounded operator $(\mathscr{A},\mathscr{D}(\mathscr{A}))$, defined by
\begin{subequations}
\begin{align}
(\mathscr{A}\bm{\zeta}) & \triangleq - \mathbf{\Lambda}\dpd{\bm{\zeta}(\xi)}{\xi} + (\mathscr{K}_2\bm{\zeta}), \\
\mathscr{D}(\mathscr{A}) & \triangleq \Bigl\{\bm{\zeta} \in H^1((0,1);\mathbb{R}^2) \mathrel{\Big|} \bm{\zeta}(0) = \bm{0}\Bigr\},
\end{align}
\end{subequations}
generates an exponentially stable $C_0$-semigroup on $L^2((0,1);\mathbb{R}^2)$. Therefore, $\tilde{\bm{X}}(t) \to \bm{0}$ exponentially fast implies the exponential convergence of $\norm{\tilde{\bm{z}}(\cdot,t)}_{L^2((0,1);\mathbb{R}^2)}$.
\end{proof}
\end{lemma}

\subsection{Frequency-domain analysis}\label{sect:freq2}
The operator $(\mathcal{H}*\tilde{\bm{Y}})(t)$ in \eqref{eq:ObsODE} should be designed to compensate the coupling term in \eqref{eq:ErrorDynODE}. This may be achieved by expressing $\tilde{Y}_2(t) = \mathbf{C}_2(\mathscr{K}_1\tilde{\bm{z}})(t)$ as a function of $\tilde{\bm{X}}(t)$, and, in turn, the latter as a function of the output error $\tilde{\bm{Y}}(t)$. 

Indeed, taking the Laplace transform of \eqref{eq:ErrorDynPDE} and solving the resulting ODE equipped with BC \eqref{eq:ErrorDynBC} yields
\begin{align}\label{eq:solLaplUx}
\begin{split}
\widehat{\tilde{\bm{z}}}(\xi,s) & = \mathbf{\Gamma}(\xi,s)\mathbf{\Lambda}^{-1}(\mathscr{K}_2\widehat{\tilde{\bm{z}}})(s) + \mathbf{\Xi}(\xi,s)\widehat{\tilde{\bm{X}}}(s), \quad \xi \in [0,1],
\end{split}
\end{align}
with
\begin{subequations}
\begin{align}
\mathbf{\Phi}(\xi,\tilde{\xi},s) & \triangleq \eu^{-\mathbf{\Lambda}^{-1}s(\xi-\tilde{\xi})}, \\
\mathbf{\Gamma}(\xi,s) & \triangleq \int_0^\xi \mathbf{\Phi}\bigl(\xi,\xi^\prime,s\bigr)\dif \xi^\prime, \\
\mathbf{\Xi}(\xi,s) & \triangleq \mathbf{\Gamma}(\xi,s)\mathbf{\Lambda}^{-1}\mathbf{A}_2.
\end{align}
\end{subequations}
where $\mathbf{\Gamma}(\xi,\cdot), \mathbf{\Xi}(\xi,\cdot) \in H^\infty(\mathbb{C}_{>0};\mathbf{M}_{2}(\mathbb{C}))$ for all $\xi \in [0,1]$. 
In turn, manipulating \eqref{eq:solLaplUx} and solving for $(\mathscr{K}_2\widehat{\tilde{\bm{z}}})(s)$ gives
\begin{align}\label{eq:K2zetaTOX}
(\mathscr{K}_1\widehat{\tilde{\bm{z}}})(s) = \mathbf{H}_1(s)\widehat{\tilde{\bm{X}}}(s),
\end{align}
with
\begin{align}
\mathbf{H}_1(s) \triangleq \begin{bmatrix}\mathbf{I}_2 & \mathbf{0} \end{bmatrix}\begin{bmatrix}\mathbf{I}_2 & -\mathbf{\Theta}_1(s) \\ \mathbf{0} & \mathbf{I}_2-\mathbf{\Theta}_2(s) \end{bmatrix}^{-1}\begin{bmatrix}\mathbf{\Psi}_1(s) \\ \mathbf{\Psi}_2(s) \end{bmatrix},
\end{align}
where
\begin{align}
\mathbf{\Theta}_1(s) & \triangleq (\mathscr{K}_1\mathbf{\Gamma})(s)\mathbf{\Lambda}^{-1}, && \mathbf{\Theta}_2(s)  \triangleq (\mathscr{K}_2\mathbf{\Gamma})(s)\mathbf{\Lambda}^{-1},  \nonumber\\
\mathbf{\Psi}_1(s) & \triangleq(\mathscr{K}_1\mathbf{\Xi})(s), && \mathbf{\Psi}_2(s)  \triangleq(\mathscr{K}_2\mathbf{\Xi})(s).
\end{align}
From the conditions on $\mathbf{\Gamma}(\xi,s)$ and $\mathbf{\Xi}(\xi,s)$, it also follows that $\mathbf{H}_1 \in H^\infty(\mathbb{C}_{>0};\mathbf{M}_{2}(\mathbb{C}))$. Consequently, combining \eqref{eq:K2zetaTOX} with the Laplace transform of \eqref{eq:errorY} yields
\begin{align}\label{eq:XtoY}
\widehat{\tilde{\bm{X}}}(s) = \mathbf{C}^{-1}(s)\widehat{\tilde{\bm{Y}}}(s),
\end{align}
being
\begin{align}\label{eq:Mc}
\mathbf{C}(s) \triangleq \begin{bmatrix} \mathbf{C}_1 \\ \mathbf{C}_2\mathbf{H}_1(s)\end{bmatrix}.
\end{align}
According to Lemma \ref{lemma:stablePoles} below, \eqref{eq:XtoY} permits reconstructing the lumped state error $\tilde{\bm{X}}(t)$ from the available measurement. 
\begin{lemma}\label{lemma:stablePoles}
Consider the matrix $\mathbf{C}(s) \in\mathbf{M}_2(\mathbb{C})$ in \eqref{eq:Mc}. Its inverse $\mathbf{C}^{-1}(s)$ has no unstable poles.

\begin{proof}
See Appendix \ref{app:1}.
\end{proof}
\end{lemma}

It is worth observing that, albeit containing no unstable poles, the transfer function $\mathbf{C}^{-1}$ is not proper, which makes it desirable to filter the output $\tilde{\bm{Y}}(t)$ to achieve a state-space realization for the operator $(\mathcal{H}*\tilde{\bm{Y}})(t)$. In particular, inspection of $\mathbf{C}^{-1}(s)$ reveals that any filter of the type
\begin{align}\label{eq:filter}
\varpi(s) = \dfrac{\gamma}{s + \gamma},
\end{align}
for arbitrary $\gamma \in \mathbb{R}_{>0}$, may be used (higher-order filters may also be employed).

\subsection{Gain injection design and output-feedback control law}
The following injection gain is considered:
\begin{align}\label{eq:injection}
\widehat{\mathcal{H}}(s) & =-\varpi(s)\bigl[\eta\mathbf{I}_2 + \mathbf{G}\mathbf{H}_1(s)\bigr]\mathbf{C}^{-1}(s),
\end{align}
where, with standard notation, $\widehat{\mathcal{H}}(s) \triangleq (\mathcal{L}\mathcal{H})(s)$, the filter $\varpi(s)$ reads according to \eqref{eq:filter}, and $\eta \in \mathbb{R}_{>0}$ is a constant gain.
With the above notation, it is possible to enounce the following result.
\begin{theorem}\label{th:2}
Consider the error dynamics \eqref{eq:ErrorDyn} with IC $(\tilde{\bm{X}}_0, \tilde{\bm{z}}_0) \in \mathcal{Y}$ satisfying the BC \eqref{eq:ErrorDynBC}, along with the injection gain \eqref{eq:injection}, with $\gamma,\eta \in \mathbb{R}_{>0}$ selected such that the transfer function
\begin{align}
\mathbf{H}_2(s) \triangleq \dfrac{s}{s^2 + \gamma s + \eta \gamma}
\end{align}
satisfies 
\begin{align}\label{eq:normCOnd}
\norm{\mathbf{H}_2(\cdot)}_\infty < \dfrac{1}{\norm{\mathbf{H}_1(\cdot)}_\infty}.
\end{align}
Then, the operator $(\mathcal{H}*\tilde{\bm{Y}})(t) = \int_0^t \mathcal{H}(t-t^\prime)\tilde{\bm{Y}}(t^\prime) \dif t^\prime$ with $\widehat{\mathcal{H}}(s) = (\mathcal{L}\mathcal{H})(s)$ as in \eqref{eq:injection} ensures that $\tilde{\bm{X}}(t), \mathbf{G}(\mathscr{K}_1\tilde{\bm{z}})(t) \to \bm{0}$ exponentially fast.

\begin{proof}
First, it may be observed that $\norm{\mathbf{H}_2(\cdot)}_\infty = 1/\gamma$, so that the inequality \eqref{eq:normCOnd} may be ensured by choosing $\gamma \in \mathbb{R}_{>0}$ sufficiently large.
Furthermore, substituting \eqref{eq:injection} into \eqref{eq:ErrorDynODE} and using \eqref{eq:K2zetaTOX} and \eqref{eq:XtoY} provides 
\begin{align}\label{eq:XtoFtilde}
\widehat{\tilde{\bm{X}}}(s) = \mathbf{H}_2(s)\mathbf{G}(\mathscr{K}_1\widehat{\tilde{\bm{z}}})(s),
\end{align}
where $\mathbf{H}_2 \in H^\infty(\mathbb{C}_{>0};\mathbf{M}_{2}(\mathbb{C}))$ for all $\gamma,\eta \in \mathbb{R}_{>0}$. In turn, inserting \eqref{eq:XtoFtilde} into \eqref{eq:K2zetaTOX} yields
\begin{align}\label{eq:XtoFtilde33}
\bigl(\mathbf{I}_2-\mathbf{H}_1(s)\mathbf{H}_2(s)\bigr)\mathbf{G}(\mathscr{K}_1\widehat{\tilde{\bm{z}}})(s) = \bm{0}.
\end{align}
The requirement \eqref{eq:normCOnd} ensures that $\norm{\mathbf{H}_1(\cdot)\mathbf{H}_2(\cdot)}_\infty < 1$, which is a sufficient condition for stability. Utilizing \eqref{eq:XtoFtilde} and \eqref{eq:XtoFtilde33}, it may be concluded that $\tilde{\bm{X}}(t), \mathbf{G}(\mathscr{K}_1\tilde{\bm{z}})(t) \to \bm{0}$ exponentially fast.
\end{proof}
\end{theorem}
Combining the assertions of Lemma \ref{lemma:2} and Theorem \ref{th:2}, the proposed observer \eqref{eq:obserStruc} guarantees the exponential stabilization of the lumped and distributed error states.
Accordingly, the main result of the paper is formalized below.
\begin{theorem}\label{th:3}
Under the same assumptions as Theorem \ref{th:2}, the observer \eqref{eq:obserStruc} with $\mathbf{M}_{2\times 1}(\mathbb{R}) \ni \mathbf{L} \triangleq [v_x\; 0]^{\mathrm{T}}$ ensures that $\norm{(\tilde{\bm{X}}(t),\tilde{\bm{z}}(\cdot,t))}_{\mathcal{X}} \to 0$ exponentially fast.
\end{theorem}

Theorem \ref{th:3} concludes the theoretical part of the paper; the next Section \ref{sect:example} exemplifies the proposed estimation strategy by adducing a numerical example.

\section{Simulation results}\label{sect:example}

The numerical values for the model parameters of the example discussed below are listed in Table \ref{tab:parameters2}. Specifically, the following numerical results refer to simulations conducted in MATLAB/Simulink\textsuperscript{\textregistered} environment. The linear PDE subsystem was solved numerically using a finite difference scheme with a discretization step of $0.02$, and combined with a time-marching algorithm with a fixed time step of $10^{-6}$ s. The ICs for the actual system were set to $\bm{X}_0 = [0.03\; -0.25]^{\mathrm{T}}$, and $\bm{z}_0(\xi) = [0.0033\; 0.0033]^{\mathrm{T}}$ (corresponding to $\norm{\bm{z}_0(\cdot)}_{L^2((0,1);\mathbb{R}^2)} = 0.0042$), whereas those for the observer to $\hat{\bm{X}}_0 = \bm{0}$ and $\hat{\bm{z}}_0(\xi) = \bm{0}$ (it is worth observing that the \emph{compatibility conditions} $\bm{z}_0(0) = \tilde{\bm{z}}_0(0)$ are not satisfied). The yaw rate and lateral acceleration measurements were contaminated with additive white noise having standard deviations of 0.01 $\text{rad}\,\text{s}^{-1}$ and 0.5 $\text{m}\,\text{s}^{-2}$, respectively, and sample times of 0.005 and 0.01 s, reflecting the levels typically observed in standard automotive sensors. 
\begin{table}[h]
\centering
\begin{tabular}{|c|c|c|c|}
\hline
\textbf{Parameter} & \textbf{Description} & \textbf{Unit} & \textbf{Value} \\
\hline 
$v_x$ & Longitudinal speed & $\textnormal{m}\,\textnormal{s}^{-1}$ & $50$ \\ 
$m$ & Vehicle mass & kg & 1300 \\ 
$I_z$ & Vertical moment of inertia  & $\textnormal{kg}\,\textnormal{m}^{2}$ & 2000 \\
$l_1$ & Front axle length & m & 1.4  \\
$l_2$ & Rear axle length & m & 1 \\
$F_{z1}$ & Front vertical force & N & $2660$ \\
$F_{z2}$ & Rear vertical force & N & $3720$ \\
$L_1$ & Front contact patch length & m & 0.1 \\
$L_2$ & Rear contact patch length & m & 0.1 \\
$\sigma_{1}$ & Front micro-stiffness & $\textnormal{m}^{-1}$ & 263 \\
$\sigma_{2}$ & Rear micro-stiffness & $\textnormal{m}^{-1}$ & 242 \\
$\phi_{1}$ & Front structural parameter & - & 0.92 \\
$\phi_{2}$ & Rear structural parameter & - & 0.92 \\
$\chi$ & Rear steering actuation & -& 0\\
\hline
\end{tabular}
\caption{Model parameters}
\label{tab:parameters2}
\end{table}

With the model parameters listed in Table \ref{tab:parameters2}, the linear single-track model described by \eqref{eq:originalSystems} is oversteer, and thus intrinsically unstable for sufficiently high longitudinal speeds, as shown in Figure \ref{fig:1} for a moderate steering maneuver with $\delta_1(t) = 1 + 2\sin(10t)$ and $\delta_2(t) = 0$. 

The observer synthesized according to the strategy outlined in Section \ref{sect:observerDesign} successfully reconstructs the lumped and distributed variables, as illustrated in Figure \ref{fig:2}, where the dynamics of the norms $\norm{(\bm{X}(t),\bm{z}(\cdot,t))}_{\mathcal{X}}$, $\norm{(\hat{\bm{X}}(t),\hat{\bm{z}}(\cdot,t))}_{\mathcal{X}}$, and $\norm{(\tilde{\bm{X}}(t),\tilde{\bm{z}}(\cdot,t))}_{\mathcal{X}}$ is visualized for $\eta=1$ and $\gamma = 500$. In particular, a satisfactory estimate is already achieved for $t \approx 0.5$ s, with the observer error converging rapidly to zero, except for some minor oscillations caused by the measurement noise. The evolution of the distributed state $\tilde{z}_1(\xi,t)$ is instead plotted in Figure \ref{ffig:3}. The conclusions that may be drawn by inspecting Figure \ref{ffig:3} are essentially the same as previously: the error $\tilde{z}_1(\xi,t)$ is stabilized close to zero, with residual oscillations introduced by the external noise. In this context, it is worth noting that the rate of convergence of the distributed states is dictated mainly by the behavior of the PDE subsystem \eqref{eq:originalSystemsPDE}, and cannot be controlled arbitrarily without measurement injection. The dynamics of the error variable $\tilde{z}_2(\xi,t)$ is similar, and not reported.

\begin{figure}
\centering
\includegraphics[width=1\linewidth]{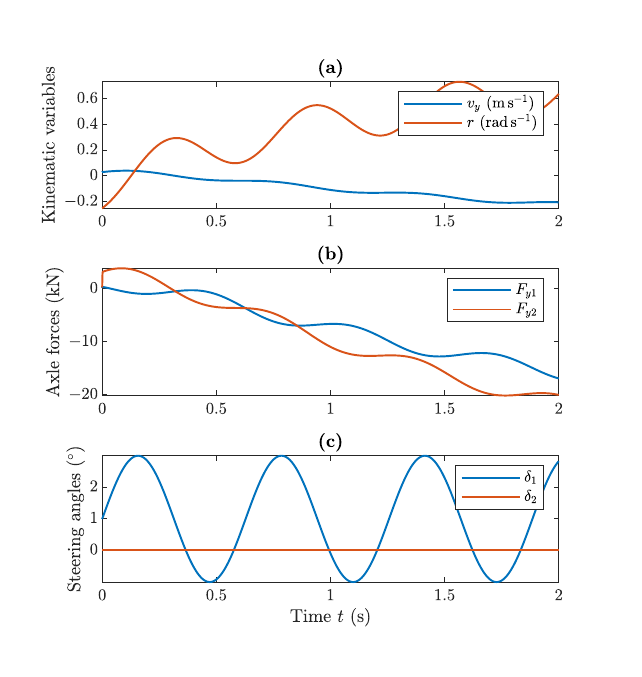} 
\caption{Behavior of the lumped states and steering inputs, for $\eta=1$ and $\gamma = 500$: \textbf{(a)} kinematic variables; \textbf{(b)} axle forces; \textbf{(c)} steering inputs.}
\label{fig:1}
\end{figure}

\begin{figure}
\centering
\includegraphics[width=1\linewidth]{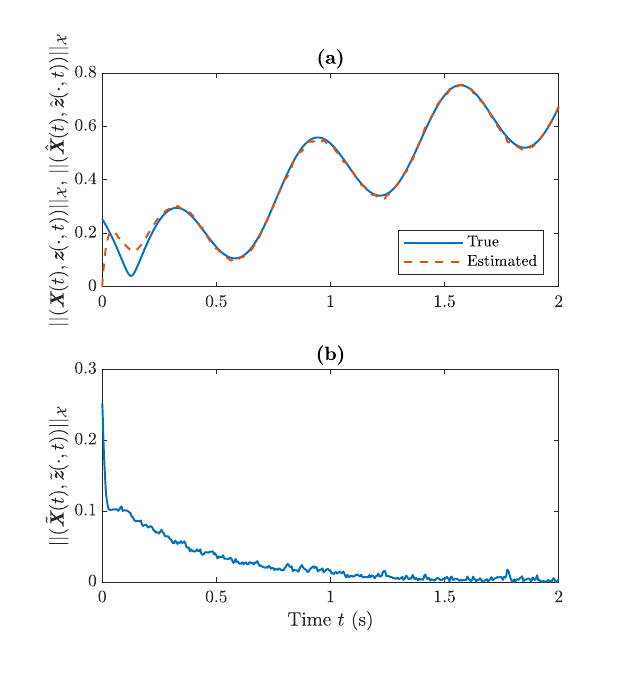} 
\caption{Dynamics of $\norm{(\bm{X}(t),\bm{z}(\cdot,t))}_{\mathcal{X}}$, $\norm{(\hat{\bm{X}}(t),\hat{\bm{z}}(\cdot,t))}_{\mathcal{X}}$, and $\norm{(\tilde{\bm{X}}(t),\tilde{\bm{z}}(\cdot,t))}_{\mathcal{X}}$ for $\eta=1$ and $\gamma = 500$.}
\label{fig:2}
\end{figure}

\begin{figure}
\centering
\includegraphics[width=0.9\linewidth]{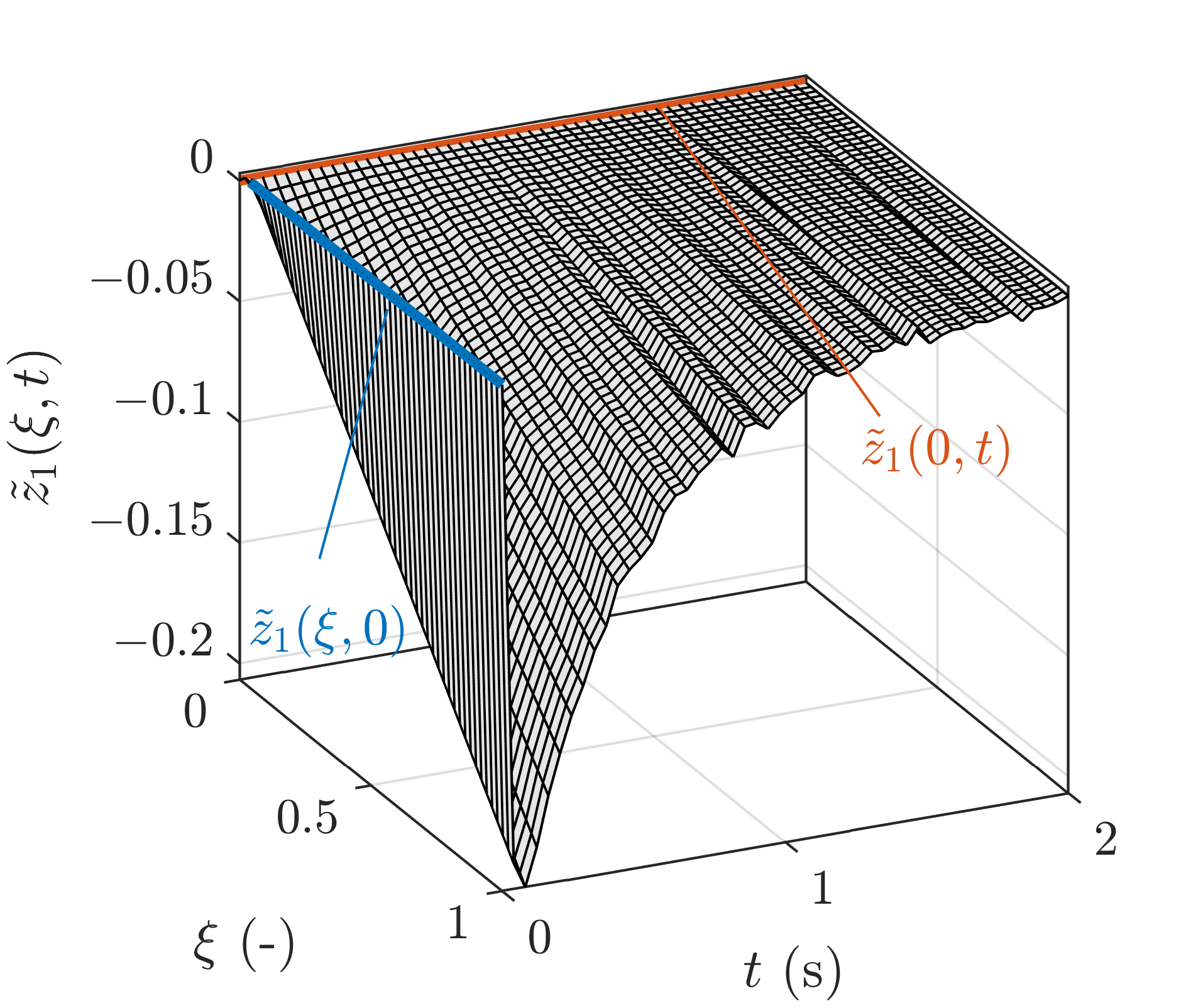} 
\caption{Evolution of the PDE state $\tilde{z}_1(\xi,t)$, along with its IC (blue line) and BC (orange line).}
\label{ffig:3}
\end{figure}

\section{Conclusions}\label{ref:concl}
An inverse-dynamics observer for a linear single-track vehicle model with distributed tire dynamics has been presented. By coupling lumped vehicle dynamics with hyperbolic PDE-based tire models, the proposed framework enables accurate estimation of both lumped and distributed states using standard sensor measurements. Theoretical analysis guarantees exponential convergence of the estimation errors through a dynamical inversion design.

Simulation results confirmed the observer's effectiveness in reconstructing the sideslip angle, tire forces, and distributed tire states, under non-ideal conditions. Future work will explore output-feedback control strategies and experimental validation on real vehicles.

\appendix
Some additional details are collected below.
\subsection{Proof of Lemma \ref{lemma:stablePoles}}\label{app:1}
This appendix contains a concise proof of Lemma \ref{lemma:stablePoles}. Lemma \ref{LemmaLambert} is propaedeutic to the result. 
\begin{lemma}[Shinozaki and Mori \cite{LambertAutomatica}]\label{LemmaLambert}
Consider the Lambert function $W : \mathbb{C} \mapsto \mathbb{C}$ with branches $W_k(\cdot)$, $k \in \mathbb{Z}$. Then, for all $(\rho,\theta) \in \mathbb{R}_{>0} \times(-\textnormal{\textpi}, \textnormal{\textpi}]$, the following inequalities hold:
\begin{subequations}
\begin{align}\label{eq:inequalityLambert}
\RE\Bigl[W_{k}\bigl(\rho\eu^{\im \theta}\bigr)\Bigr] &< \RE\Bigl[W_{-1}\bigl(\rho\eu^{\im \theta}\bigr)\Bigr], && k < -1, \\
\RE\Bigl[W_{k}\bigl(\rho\eu^{\im \theta}\bigr)\Bigr] &< \RE\Bigl[W_{1}\bigl(\rho\eu^{\im \theta}\bigr)\Bigr], && k > 1.
\end{align}
\end{subequations}
Moreover, for all $(\rho,\theta) \in (0,1/\eu) \times(-\textnormal{\textpi}, \textnormal{\textpi}]$,
\begin{align}\label{eq:inequalityLambert}
\RE\Bigl[W_{k}\bigl(\rho\eu^{\im \theta}\bigr)\Bigr] < \RE\Bigl[W_{0}\bigl(\rho\eu^{\im \theta}\bigr)\Bigr], && k \in \mathbb{Z}\setminus\{0\}.
\end{align}
\begin{proof}
See \cite{LambertAutomatica}.
\end{proof}
\end{lemma}

By invoking Lemma \ref{LemmaLambert}, it is possible to prove \ref{lemma:stablePoles}.
\begin{proof}[Proof of Lemma \ref{lemma:stablePoles}] 
Lengthy but straightforward calculations give
\begin{align}\label{eq:detC3}
\det\bigl(\mathbf{C}(s)\bigr) = F_{z1}\sigma_1\phi_1h_1(s) +  F_{z2}\sigma_2\phi_2h_2(s),
\end{align}
where
\begin{align}
h_i(s) \triangleq \dfrac{\frac{s}{\lambda_i}+\eu^{-s/\lambda_i}-1}{s^2/\lambda_i}\left[1-\psi_i\Biggl(\dfrac{1-\eu^{-s/\lambda_i}}{s/\lambda_i} \Biggr) \right]^{-1},
\end{align}
with $\lambda_i \triangleq v_x/L_i$, $i \in \{1,2\}$. First, it may be observed that
\begin{align}
\lim_{s\to 0} \lambda_i(s) = \dfrac{1}{2\lambda_i\phi_i} >0, \quad i \in \{1,2\}.
\end{align}
Moreover, by setting $\rho \triangleq 1/\eu$, the zeros of the functions $h_i \in H^\infty(\mathbb{C}_{>0};\mathbb{C})$ may be calculated as
\begin{align}\label{eq:solW+1}
s =\lambda_i\Bigl[W_{k}\bigl(\rho\eu^{\im \text{\textpi}}\bigr) + 1\Bigr], && k \in \mathbb{Z} \setminus\{-1,0\}.
\end{align}
Lemma \ref{LemmaLambert} implies $\RE(s) < 0$ for all $s$ calculated according to \eqref{eq:solW+1}, and consequently that $\RE(h_i(s)) > 0$ for all $s\in \mathbb{C}_{>0}$, $i\in \{1,2\}$. In turn, since the coefficients multiplying $h_1(s)$ and $h_2(s)$ in \eqref{eq:detC3} are both positive, this ensures that $\mathbf{C}^{-1}(s)$ has no unstable poles.
\end{proof}

\end{document}